\documentclass[12pt]{article}

\usepackage{amsmath}
\usepackage{amssymb}
\usepackage{amsthm}
\usepackage{latexsym}
\usepackage{cite}
\usepackage{psfrag}
\usepackage{epsfig}
\usepackage{graphicx}
\usepackage{color}
\usepackage{amsfonts}
\usepackage{mathrsfs}
\usepackage{url}
\usepackage{hyperref}
\usepackage{array}
\usepackage{flushend}
\usepackage{algorithm}
\usepackage{algorithmic}

\newtheorem{theorem}{Theorem}

\newtheorem{lemma}[theorem]{Lemma}
\theoremstyle{definition}

\newtheorem{example}[theorem]{Example}
\newtheorem{remark}[theorem]{Remark}
\newtheorem{notation}[theorem]{Notation}

\def\BB{\mathfrak{B}}
\def\CC{\mathcal{C}}

\def\V{\mathcal{V}}

\def\C{\mathscr{C}}
\def\M{\mathscr{M}}

\def\N{\mathscr{N}}
\def\PG{\mathrm{PG}}
\def\F{\mathbb{F}}
\def\TT{\mathbb{T}}

\def\P{\mathbf{P}}
\def\c{\mathbf{c}}
\def\v{\mathbf{v}}
\def\x{\mathbf{x}}

\begin{document}
\title{\textbf{On cosets weight distributions of the doubly-extended Reed-Solomon codes of codimension 4}
\author{Alexander A. Davydov\footnote{The research of A.A.~Davydov was done at IITP RAS and supported by the Russian Government (Contract No 14.W03.31.0019).} \\
{\footnotesize Institute for Information Transmission Problems
(Kharkevich
institute), Russian Academy of Sciences}\\
{\footnotesize Bol'shoi Karetnyi per. 19, Moscow,
127051, Russian Federation. E-mail: adav@iitp.ru}
\and Stefano Marcugini\footnote{The research of  S. Marcugini and F.~Pambianco was
 supported in part by the Italian
National Group for Algebraic and Geometric Structures and their Applications (GNSAGA - INDAM) (No U-UFMBAZ-2017-000532, 28.11.2017) and by
University of Perugia (Project: Curve, codici e configurazioni di punti, Base Research
Fund 2018, No 46459, 15.06.2018).}  \,and Fernanda Pambianco$^\dag$  \\
{\footnotesize Department of Mathematics and Computer Science, University of Perugia,}\\
{\footnotesize Via Vanvitelli~1, Perugia, 06123, Italy. E-mail:
\{stefano.marcugini,fernanda.pambianco\}@unipg.it}}
}
\date{}
\maketitle
\textbf{Abstract.} 
 We consider the $[q+1,q-3,5]_q3$ generalized doubly-extended Reed-Solomon code of codimension~$4$ as the code associated with the twisted cubic in the projective space $\mathrm{PG}(3,q)$. Basing on the point-plane incidence matrix of $\mathrm{PG}(3,q)$, we  obtain the number of weight 3 vectors in all the cosets of the considered code. This allows us to classify the cosets by their weight distributions and to obtain these distributions. The weight of a coset is the smallest Hamming weight of any vector in the coset. For the cosets of equal weight having distinct weight distributions, we prove that the difference between the $w$-th components, $3<w\le q+1$, of the distributions is uniquely determined by the difference between  the $3$-rd components.  This implies an interesting (and in some sense unexpected) symmetry of the obtained distributions.


\textbf{Mathematics Subject Classication (2010).} 94B05, 51E21, 51E22

\section{Introduction}\label{sec_Intro}
Let $\F_{q}$ be the Galois field with $q$ elements, $\F_{q}^*=\F_{q}\setminus\{0\}$, $\F_q^+=\F_q\cup\{\infty\}$. Let $\F_{q}^{n}$ be
the space of $n$-dimensional vectors over ${\mathbb{F}}_{q}$.  We denote by  $[n,k,d]_{q}R$ an $\F_q$-linear code of length~$n$, dimension $k$, minimum distance~$d$, and covering radius $R$. If $d=n-k+1$,   it is a maximum distance separable (MDS) code.
Let $\PG(N,q)$ be the $N$-dimensional projective space over $\F_q$.
 For an introduction to coding theory see \cite{Blahut,MWS,Roth}.  For an introduction to projective spaces over finite fields and connections between  projective geometry and coding theory see \cite{Hirs_PGFF,HirsSt-old,HirsStor-2001,HirsThas-2015,LandSt,EtzStorm2016}.

In an $(N+1)\times(q+1)$ parity check matrix of a $[q+1,q-N,N+2]_q$ generalized doubly-extended
Reed-Solomon (GDRS) code, the $j$-th column has the form $(v_j,v_j\alpha_j,v_j\alpha_j^2,\ldots,v_j\alpha_j^N)^{tr}$, where $j=1,2,\ldots,q$; $\alpha_1,\ldots,\alpha_q$ are distinct elements of $\F_q$; $v_1,\ldots,v_q$ are nonzero (not necessarily distinct) elements of $\F_q$. Also, this matrix contains one more column $(0,\ldots,0,v)^{tr}$ with $v\neq 0$.
 The code, dual to a GDRS code, is a GDRS code too.  Clearly, a GDRS code is MDS.

An $n$-arc in  $\PG(N,q)$, with $n\ge N + 1\ge3$, is a
set of $n$ points such that no $N +1$ points belong to
the same hyperplane of $\PG(N,q)$. An $n$-arc is complete if it is not contained in an $(n+1)$-arc. Arcs and MDS codes are equivalent objects, see e.g. \cite[Sec. 11.6]{MWS}, \cite[Sec. 3.1]{LandSt}, \cite[Sec. 2.1]{EtzStorm2016}.

In $\PG(N,q)$, $2\le N\le q-2$, a normal rational curve is any $(q+1)$-arc projectively equivalent to the arc
$\{(1,t,t^2,\ldots,t^{N-1}, t^N):t\in \F_q\}\cup \{(0,\ldots,0 ,1)\}$.  The points (in homogeneous coordinates) of a normal rational curve in $\PG(N,q)$
treated as columns define a parity check matrix of a $[q + 1,q-N,N + 2]_q$ GDRS code \cite{EtzStorm2016,LandSt,Roth}. We say that this GDRS code is \emph{associated} with the normal rational curve. In $\PG(3,q)$, the normal rational curve is called a  \emph{twisted cubic} \cite{Hirs_PG3q,HirsThas-2015}. We denote the cubic by $\C$, see Section \ref{subsec_twis_cub} for preliminaries.

Let $\CC_\C$ be the $[q+1,q-3,5]_q3$ GDRS code associated with the cubic $\C$.

Twisted cubics in $\PG(3,q)$ have been widely studied; see e.g. \cite{BrHirsTwCub,Hirs_PG3q,GiulVincTwCub,BDMP_TwistCubArX} and the references therein.  In particular, in \cite{Hirs_PG3q}, the orbits of planes and points under the group $G_q$ of the projectivities fixing a cubic are considered. In \cite{BDMP_TwistCubArX}, the structure of the point-plane incidence matrix  in $\PG(3,q)$ with respect to the orbits of points and planes under $G_q$ is described; see Section~\ref{subsec_incid} for useful details.

Also in \cite{BDMP_TwistCubArX}, it is shown that  twisted cubics can be treated as  multiple $\rho$-saturating sets with $\rho=2$ which, in turn, give rise to  asymptotically optimal  non-binary linear multiple covering $[q+1,q-3,5]_q3$ codes of radius $R=3$. This means that the code $\CC_\C$ can be viewed as an asymptotically optimal  multiple covering. For an introduction to multiple covering codes and multiple saturating sets see \cite{CHLL_CovCodBook,BDGMP_MultCov,BDGMP_MultCov2016}.

A \emph{coset} of a code is a translation of the code. A coset $\V$ of an $[n,k,d]_{q}R$ code $\CC$ can be represented as
\begin{align}\label{eq1_coset}
  \V=\{\x\in\F_q^n\,|\,\x=\c+\v,\c\in \CC\}\subset\F_q^n
\end{align}
 where $\v\in \V$ is a vector  fixed for the given representation and $\c$ is a codeword; see e.g. \cite{Blahut,HufPless,MWS,HandbookCodes,Roth} and the references therein. For preliminaries, see Section \ref{subsec_cosets}.

 The weight distribution of code cosets, including the number of the cosets with distinct distributions, is interesting by itself; it is an important combinatorial property of a code. In particular, the distribution serves to estimate decoding results.  There are many papers connected with distinct aspects of the weight distribution of cosets for codes over distinct fields and rings, see e.g. \cite{Blahut}, \cite[Sec. 5.5, 6.6, 6.9]{MWS}, \cite[Sec. 7]{HufPless}, \cite[Sec. 10]{HandbookCodes}, \cite{AsmMat,Bonneau1990,BonnDuursma,CharpHelZin,CheungIEEE1989,CheungIEEE1992,DMP_Weight1_2Cosets,Delsarte4Fundam,DelsarteLeven,Helleseth,%
JurrPellik,KaipaIEEE2017,KasLin,MW1963,Schatz1980,ZDK_DHRIEEE2019}, \cite[Sec. 6.3]{DelsarteBook}, and the references therein. Nevertheless,  as far as it is known to the authors, the weight distribution of GDRS code cosets is an open problem.

\emph{In this paper}, we consider the weight distribution of the cosets of the $[q+1,q-3,5]_q3$ GDRS code of codimension $4$. We consider it as the code $\CC_\C$ associated with the twisted cubic. By a geometrical way, basing on the $\PG(3,q)$ point-plane incidence matrix of \cite{BDMP_TwistCubArX}, we  obtain the number of weight 3 vectors in all the cosets of $\CC_\C$. This allows us to use the results  of the paper \cite{Bonneau1990} that connect the full weight distribution of a coset of an $[n,k,d]_q$ MDS code with the first $d-2$ components of the distribution. As a result, we obtain exact formulas for the weight distribution of all the cosets of the $[q+1,q-3,5]_q3$ GDRS code and classify the cosets of $\CC_\C$ by their weight distributions.

 The weight of a coset is the smallest Hamming weight of any vector in the coset. For the cosets of equal weight having distinct weight distributions, we prove the following  \emph{property of differences}: the difference between  the $w$-th components, $3<w\le q+1$, of the distributions is uniquely determined by the difference between  the $3$-rd components.  This implies an interesting (and in some sense unexpected) \emph{symmetry of the weight distributions of the cosets of the $[q+1,q-3,5]_q3$ code}.

 Note that the weight 3 cosets of $\CC_\C$ contain the so-called farthest-off points or deep-holes. The investigation of the deep-holes of Reed-Solomon codes is of great interest, see e.g. \cite{KaipaIEEE2017,ZDK_DHRIEEE2019} and the references therein.

 This paper is organized as follows. Section \ref{sec_prelimin} contains preliminaries. In Section \ref{sec_distrib3}, we give the distribution of weight 3 vectors in the cosets of the code $\CC_\C$ and classify the cosets. In Section~\ref{sec_wd_coset}, the exact formulas for the weight distribution of all the cosets of the $[q+1,q-3,5]_q3$ GDRS code are obtained.  In Section \ref{sec_symmetry}, we prove the property of differences and consider a symmetry in the weight distributions of the cosets of the code $\CC_\C$.

Throughout the paper, we consider $q\ge5$.

\section{Preliminaries}\label{sec_prelimin}

\subsection{Twisted cubic}\label{subsec_twis_cub}
In this subsection, we summarize some known results on twisted cubics from \cite[Ch. 21]{Hirs_PG3q}.

Let $\P(x_0,x_1,x_2,x_3)$ be a point of $\PG(3,q)$ with the homogeneous coordinates $x_i\in\F_{q}$; the leftmost nonzero coordinate is equal to $1$. For $t\in\F_q^+$, let  $P(t)$ be a point such that
\begin{align*}
&P(t)=\P(1,t,t^2,t^3)\text{ if }t\in\F_q,~~P(\infty)=\P(0,0,0,1).
\end{align*}
Let $\C\subset\PG(3,q)$ be the \emph{twisted cubic} consisting of $q+1$ points $P_1,\ldots,P_{q+1}$ no four of which are coplanar.
We consider $\C$ in the canonical form
\begin{align}\label{eq2_cubic}
\C=\{P_1,P_2,\ldots,P_{q+1}\}=\{P(t)|t\in\F_q^+\}.
\end{align}
Let $\boldsymbol{\pi}(c_0,c_1,c_2,c_3)\subset\PG(3,q)$, $c_i\in\F_q$, be the plane with equation
$c_0x_0+c_1x_1+c_2x_2+c_3x_3=0.$

In every point $P(t)\in\C$, there is an \emph{osculating plane} $\pi_\text{osc}(t)$ such that
 \begin{align*}
  \pi_\text{osc}(t)=\boldsymbol{\pi}(-t^3,3t^2,-3t, 1)\text{ if }t\in\F_q; ~\pi_\text{osc}(\infty)=\boldsymbol{\pi}(0,0,0,1);~P(t)\in\pi_\text{osc}(t).
 \end{align*}
In addition, $P(t)$ is the unique point of $\C$ belonging to $\pi_\text{osc}(t)$, $t\in\F_q^+$.
In total, there are $q+1$ osculating planes; they form the osculating developable $\Gamma$ to $\C$, that is, a \emph{pencil of planes} for $q\equiv0\pmod3$ or a \emph{cubic developable} for $q\not\equiv0\pmod3$. For $q\equiv0\pmod3$, the intersection line of all the osculating planes is the axis of the pencil;    every point of the axis lies on exactly $q+1$ osculating planes and exactly one tangent to $\C$.

A \emph{chord} of $\C$ is a line through a pair of real points of $\C$ or a pair of complex conjugate points. In the last  case it is an \emph{imaginary chord}. If the real points are distinct, it is a \emph{real chord}. If the real points coincide with each other, it is a \emph{tangent}  to $\C$.

\begin{notation}\label{notation_1}
The following notation is used:
\begin{align*}
  &G_q && \text{the group of projectivities in } \PG(3,q) \text{ fixing }\C;\displaybreak[3] \\
  &\Gamma &&\text{the osculating developable to }  \C;\displaybreak[3]\\
&\Gamma\text{-plane}  &&\text{an osculating plane of }\Gamma;\displaybreak[3]\\
&d_\C\text{-plane}&&\text{a plane containing \emph{exactly} $d$ distinct points of }\C,~d=0,1,2,3;\displaybreak[3]\\
&1_{\C}\setminus\Gamma\text{-plane}&&\text{a $1_\C$-plane not in $\Gamma$;}\displaybreak[3]\\
&\C\text{-point}&&\text{a point  of $\C$;}\displaybreak[3]\\
&\mu_\Gamma\text{-point}&&\text{a point off $\C$ lying on \emph{exactly} $\mu$ osculating planes, } \displaybreak[3]\\
&&&\mu_\Gamma\in\{0_\Gamma,1_\Gamma,3_\Gamma\}\text{ for }q\not\equiv0\pmod 3,\displaybreak[3]\\
&&&\mu_\Gamma\in\{(q+1)_\Gamma\}\text{ for }q\equiv0\pmod 3;\displaybreak[3]\\
&\text{T-point}&&\text{a point  off $\C$  on a tangent to $\C$ for $q\not\equiv0\pmod 3$;}\displaybreak[3]\\
&\text{TO-point}&&\text{a point  off $\C$ on a tangent and one osculating plane for }q\equiv 0\pmod 3;\displaybreak[3]\\
&\text{RC-point}&&\text{a point  off $\C$  on a real chord;}\displaybreak[3]\\
&\text{IC-point}&&\text{a point  on an imaginary chord;}\displaybreak[3]\\
&\#S&&\text{the cardinality of a set }S.
\end{align*}
\end{notation}

\begin{remark}\label{rem2:afterNot1}
For $q\not\equiv0\pmod 3$, the $\mu_\Gamma$-points with $\mu_\Gamma\in\{0_\Gamma,1_\Gamma,3_\Gamma\}$ can be viewed also as RC- or IC-points, see \eqref{eq2_=1_orbit_point} and \eqref{eq2_=1b_orbit_point} in Theorem \ref{th2_HirsCor4,5}(ii) below. Namely such a treatment of the $\mu_\Gamma$-points is used in this paper, see Remark \ref{rem2:r_ij} and Theorem \ref{th3_3_coset} below.

For $q\equiv0\pmod 3$, all the $q+1$ points of the axis of the $\Gamma$-plane pencil are the $\mu_{q+1}$-points and vice versa. Also, the axis intersects every tangent so that every point of the axis lies on exactly one tangent; among the remaining $q$ points one  belongs to $\C$ and $q-1$ are TO-points.
\end{remark}

\begin{theorem}\label{th2_HirsCor4,5}
\emph{\cite[Ch. 21]{Hirs_PG3q}} Under $G_q$, there are five orbits $\N_i$ of planes and five orbits $\M_j$ of points. These orbits have the following properties:

\textbf{\emph{(i)}} For all $q$, the orbits $\N_i$ of planes are as follows: $\N_1=\{\Gamma\text{-planes}\},~\N_{2}=\{2_\C\text{-planes}\}$,
$\N_{3}=\{3_\C\text{-planes}\}, ~ \N_{4}=\{1_\C\setminus\Gamma\text{-planes}\},~\N_{5}=\{0_\C\text{-planes}\}. $

\textbf{\emph{(ii)}} For $q\not\equiv0\pmod 3$, the orbits $\M_j$ of points are as follows:
\begin{align}
&\M_1=\C,\, \M_2=\{\text{\emph{T}}\text{-points}\},\,\M_3=\{3_\Gamma\text{-points}\},\, \M_4=\{1_\Gamma\text{-points}\},\,\M_5=\{0_\Gamma\text{-points}\};\notag\displaybreak[3]\\
&\#\M_1=q+1,\,\#\M_{2}=q(q+1),\,\#\M_{3}=\frac{q^3-q}{6},\,\#\M_{4}=\frac{q^3-q}{2},\,\#\M_{5}=\frac{q^3-q}{3}.\displaybreak[3]\label{eq2_point_orbits_gen}\\
&\text{Also,}\notag\displaybreak[3]\\
 &\text{if } q\equiv1\pmod 3 \text{ then } \M_{3}\cup\M_{5}=\{\text{\emph{RC}-points}\}, ~\M_{4}=\{\text{\emph{IC}-points}\};\displaybreak[3]\label{eq2_=1_orbit_point}\\
 &\text{if } q\equiv-1\pmod 3\text{ then }\M_{3}\cup\M_{5}=\{\text{\emph{IC}-points}\},~
 \M_{4}=\{\text{\emph{RC}-points}\}.\label{eq2_=1b_orbit_point}
\end{align}

\textbf{\emph{(iii)}} For $q\equiv0\pmod 3$, the orbits $\M_k$ of points are as follows:
\begin{align}
&\M_1=\C,~\M_2=\{(q+1)_\Gamma\text{-points}\},~\M_3=\{\text{\emph{TO}-points}\},~\M_4=\{\text{\emph{RC}-points}\},\displaybreak[3]\notag\\
&\M_5=\{\text{\emph{IC}-points}\};\,\#\M_1=\#\M_2=q+1; \,\#\M_3=q^2-1;\,\#\M_4=\#\M_5=\frac{q^3-q}{2}.\label{eq2_point_orbits_j=0}
\end{align}

\textbf{\emph{(iv)}} No two chords of $\C$ meet off $\C$. Every point off $\C$ lies on exactly one chord of $\C$.
\end{theorem}

Note that the property ``every point off $\C$ lies on exactly one chord of $\C$'' is important for further researches in this paper, in particular, for Theorem \ref{th3_3_coset}.

 \subsection{The number of $3_\C$-planes through points and lines of $\PG(3,q)$}\label{subsec_incid}
\begin{lemma}\label{lem2_3 2} \emph{\cite[Lem. 4.3]{BDMP_TwistCubArX}}
 The number of\/ $3_\C$-planes and $2_\C$-planes through a real chord of $\C$ is equal to $q-1$ and $2$, respectively.
\end{lemma}
\begin{lemma}\emph{\cite[Lem. 4.12]{BDMP_TwistCubArX}}
Through every point of the orbit $\M_j$ we have the same number of planes from the orbit $\N_i$.
\end{lemma}

Throughout the paper, for $q\equiv \xi\pmod3$, let $r_{3j}^{(\xi)}$ be the number of $3_\C$-planes from the orbit $\N_3$ through every point of the orbit $\M_j$.

 \begin{theorem}\emph{\cite[Th. 3.1, Tab. 1, 2]{BDMP_TwistCubArX}}\label{Th2_rij}
Let  $q\equiv \xi\pmod3$. Let  $r_{3j}^{(\xi)}$ be as above. The following holds:
\begin{align*}
&\textbf{\emph{(i)}}~~r_{31}^{(1)}=r_{31}^{(-1)}=\frac{1}{2}(q^2-q),~r_{32}^{(1)}=r_{32}^{(-1)}=\frac{1}{6}(q^2-3q+2);\displaybreak[3]\\
&r_{33}^{(1)}=\frac{1}{6}(q^2+q+4),~r_{34}^{(1)}=\frac{1}{6}(q^2-q),~r_{35}^{(1)}=\frac{1}{6}(q^2+q-2);\displaybreak[3]\\
&r_{33}^{(-1)}=\frac{1}{6}(q^2- q+4),~r_{34}^{(-1)}=\frac{1}{6}(q^2+q),~r_{35}^{(-1)}=\frac{1}{6}(q^2-q-2).\displaybreak[3]\\
&\textbf{\emph{(ii)}}~~r_{31}^{(0)}=\frac{1}{2}(q^2-q),~r_{32}^{(0)}=r_{35}^{(0)}=\frac{1}{6}(q^2-q),~r_{33}^{(0)}=\frac{1}{6}(q^2-3q),~r_{34}^{(0)}=\frac{1}{6}(q^2+q).
\end{align*}
 \end{theorem}

\begin{remark}\label{rem2:r_ij}
 Let  $q\equiv \xi\pmod3$.

  The points off $\C$ lying on a real chord (RC-points) form the orbits $\M_3$ and $\M_5$ for $\xi=1$ or the orbit $\M_4$ for $\xi\ne1$, see Theorem \ref{th2_HirsCor4,5}(ii)(iii). By Theorem~\ref{Th2_rij}, if $\xi=1$, we have $r_{33}^{(1)}=(q^2+q+4)/6$ and $r_{35}^{(1)}=(q^2+q-2)/6$ $3_\C$-planes through every point of the orbits $\M_3$ and $\M_5$, respectively. On the other hand, if $\xi\ne1$, there are $r_{34}^{(-1)}=r_{34}^{(0)}=(q^2+q)/6$ $3_\C$-planes through every point of~$\M_4$.  The RC-points are connected with the case (3) of Proof of Theorem \ref{th3_3_coset}.

  The points off $\C$ lying on a tangent are either T-points (orbit $\M_2$ for $\xi\ne0$), or $(q+1)_\Gamma$-points (orbit $\M_2$ for $\xi=0$), or, finally,  TO-points (orbit $\M_3$ for $\xi=0$); see Theorem \ref{th2_HirsCor4,5}(ii)(iii) and Remark~\ref{rem2:afterNot1}. By Theorem~\ref{Th2_rij}, if $\xi\ne0$, we have $r_{32}^{(1)}=r_{32}^{(-1)}=(q^2-3q+2)/6$ $3_\C$-planes through every T-point whereas, if $\xi=0$,  there are $r_{32}^{(0)}=(q^2-q)/6$ and $r_{33}^{(0)}=(q^2-3q)/6$ $3_\C$-planes through every $(q+1)_\Gamma$- and TO-point, respectively; cf. the values $B_3(\V)$ for the cosets $\V^{(3)}_a$ in Theorem \ref{th3_3_coset}. For T- and TO-points, the values noted are $2,5,7,9,15$ for $q=5,7,8,9,11$, respectively; cf. values of $B_3=B_3(\V)$ in the cosets $\V^{(3)}_a$ for these $q$ in Example \ref{ex4} below.

  The points off $\C$ lying on an imaginary chord (IC-points) form either the orbits $\M_3$ and $\M_5$ for $\xi=-1$ or the orbit $\M_4$ for $\xi=1$, or, finally, the orbit $\M_5$ for $\xi=0$, see Theorem \ref{th2_HirsCor4,5}(ii)(iii).

  By Theorem~\ref{Th2_rij}, if $\xi=-1$, we have $r_{33}^{(-1)}=(q^2-q+4)/6$ and $r_{35}^{(-1)}=(q^2-q-2)/6$ $3_\C$-planes through every point of the orbits $\M_3$ and $\M_5$, respectively; cf. relations $B_3(\V)$ for the cosets $\V^{(3)}_c$ and $\V^{(3)}_b$ with $\xi=-1$ in Theorem \ref{th3_3_coset}. The values noted are $\{4,3\}$, $\{10,9\}$, and $\{19,18\}$ for $q=5,8$, and $11$, respectively; cf. values of $B_3=B_3(\V)$ in the cosets $\{\V^{(3)}_c,\V^{(3)}_b\}$ for these $q$ in Example~\ref{ex4}.

  On the other hand, if $\xi\ne-1$, there are $r_{34}^{(1)}=r_{35}^{(0)}=(q^2-q)/6$ $3_\C$-planes through every point of~$\M_4$ (for $\xi=1$) or $\M_5$ (for $\xi=0$); cf. relations $B_3(\V)$ for the cosets $\V^{(3)}_b$ with $\xi\ne-1$ in Theorem \ref{th3_3_coset}. The value noted is $7,12$ for $q=7,9$, respectively; cf. values of $B_3=B_3(\V)$ in the cosets $\V^{(3)}_b$ for these $q$ in Example \ref{ex4}.

 The T- and IC-points are connected with the case (4) of Proof of Theorem \ref{th3_3_coset}.
\end{remark}


\subsection{Cosets of a linear code}\label{subsec_cosets}
We give a few known definitions and properties connected with cosets of linear codes, see e.g. \cite{Blahut,HufPless,MWS,HandbookCodes,Roth} and the references therein.

We consider a coset $\V$ of an $[n,k,d]_{q}R$ code $\CC$ in the form \eqref{eq1_coset}. We have $\#\V=\#\CC=q^k$. One can take as $\v$ any vector of $\V$. So, there are $\#\V=q^k$ formally distinct representations of the form~\eqref{eq1_coset}; they all give the same coset $\V$. If $\v\in \CC$, we have $\V=\CC$. The distinct cosets of~$\CC$ partition $\F_{q}^{n}$ into $q^{n-k}$ sets of size $q^k$.
\begin{notation}\label{notation_coset}
For an $[n,k,d]_{q}R$ code $\CC$ and its coset $\V$ of the form \eqref{eq1_coset}, the following notation is used:
\begin{align*}
&t=\left\lfloor\frac{d-1}{2}\right\rfloor&&\text{the number of correctable errors};\displaybreak[3]\\
&wt(\x)&&\text{the Hamming weight of a vector  $\x\in\F_q^n$;}\displaybreak[3]\\
&A_w(\CC)&&\text{the number of weight $w$ codewords of $\CC$;}\displaybreak[3]\\
&S(\CC)&&\text{the set of nonzero weights in $\CC$; }~S(\CC)=\{w>0|A_w(\CC)\ne0\};\displaybreak[3]\\
&s(\CC)=\#S(\CC)&&\text{the number of nonzero weights in }\CC;\displaybreak[3]\\
&\CC^\bot&&\text{the $[n,n-k,d^\bot]_qR^\bot$ code dual to $\CC$;}\displaybreak[3]\\
&\v+\CC&&\text{the coset of $\CC$ of the form \eqref{eq1_coset}};\displaybreak[3]\\
&B_w(\V)&&\text{the number of weight $w$ vectors in }\V;\displaybreak[3]\\
&\text{the weight of a coset}&&\text{the smallest Hamming weight of any vector in the coset;}\displaybreak[3]\\
&\V^{(W)}&&\text{a coset of weight }W;~~\text{ if $w<W$ then } A_w(\V^{(W)})=0;\displaybreak[3]\\
&\text{a coset leader}&&\text{a vector in the coset having the smallest Hamming weight;}\displaybreak[3]\\
&H(\CC)&&\text{an $(n-k)\times n$ parity check matrix of }\CC;\displaybreak[3]\\
&\x^{tr}&&\text{the transposed vector }\x;\displaybreak[3]\\
&H(\CC)\x^{tr}&&\text{the \emph{syndrome} of a vector }\x\in\F_q^n,~~H(\CC)\x^{tr}\in \F_q^{n-k};\displaybreak[3]\\
&\text{a coset syndrome}&&\text{the syndrome of any vector of the coset}.
\end{align*}
\end{notation}

In cosets of weight $> t$, a  vector of the minimal weight is not necessarily unique. Cosets of weight $\leq t$ have a unique leader.

 The code $\CC$  is the coset of weight zero.  The leader of $\CC$ is the zero vector of $ \F_q^{n}$.

All vectors in a code coset have the same syndrome; it is called the \emph{coset syndrome}. Thus, there is a one-to-one correspondence between cosets and syndromes.
 The syndrome of $\CC$ is the zero vector of $ \F_q^{n-k}$.

A linear $[n,k]_{q}$ code has covering radius $R$ if
every vector of $\F_{q}^{n-k}$ considered as a column is equal to a linear combination of at most $R$ columns
of a parity check matrix of the code, and $R$ is the smallest value with
such property.

 The covering radius $R$ of the code $\CC$ is equal to the maximum weight of a coset of $\CC$.

 \begin{theorem}\label{th2_HP_coset}
\begin{description}
  \item[(i)] \emph{\cite[Lem. 7.5.1]{HufPless}}  For a code $\CC$, the weight distributions of all cosets $\alpha\v+\CC$, $\alpha\in\F_q^*$, are identical.
  \item[(ii)] \emph{\cite[Th. 6.20]{MWS}}, \emph{\cite[Th. 7.5.2]{HufPless}}, \emph{\cite[Th. 10.10]{HandbookCodes}}
  For a code $\CC$, the weight distribution of any coset of weight $<s(\CC^\bot)$ is uniquely determined if, in the coset, the numbers
of vectors of weights $1,2,\ldots,s(\CC^\bot)-1$ are known.
\end{description}
\end{theorem}

\begin{theorem}\label{th2_Bon} \emph{\cite{Bonneau1990}}
Let $\CC$ be an $[n,k,d=n-k+1]_q$ MDS code. Let $\V$ be one of its cosets.  Assume that all the values of $B_v(\V)$ with $0\le v\le d-2$ are known. Then, for $w\ge d-1$, the weight distribution of\/ $\V$ is as follows:
\begin{align}
 B_{w}(\V)=\binom{n}{w}\sum_{j=0}^{w-d+1}(-1)^j\binom{w}{j}q^{w-d+1-j}+\BB_{w}(\V),~w\ge d-1,\label{eq2_wd_cosetBon}
\end{align}
where
\begin{align}
&\BB_{w}(\V)=\sum_{j=w-d+2}^w(-1)^j\sum_{v=0}^{w-j}\binom{j+n-w}{j}\binom{n-v}{w-j-v}B_v(\V).\label{eq2_Bw2Bon}
\end{align}
\end{theorem}

As far as it is known to the authors, the weight distribution of cosets of GDRS codes, including the number of the cosets with distinct weight distributions,  is an open combinatorial problem.

\section{The distribution of weight 3 vectors in the cosets of the code $\CC_\C$. Classification of the cosets}\label{sec_distrib3}
\begin{notation}
Let $\CC_\C$  be the $[q+1,q-3,5]_q3$  GDRS code \emph{associated} with the twisted cubic $\C$ of \eqref{eq2_cubic} so that columns of its parity check matrix  are points of $\C$ in homogeneous coordinates.
\end{notation}

In $PG(3,q)$, every point off $\C$ lies in a few $3_\C$-planes, see Section \ref{subsec_incid}; therefore, the covering radius of $\CC_\C$ is equal to 3 and the cosets of $\CC_\C$ have weight $\le3$.

\begin{remark}\label{rem3}
Every point $P$ of $\PG(3,q)$ gives rise to $q-1$ nonzero syndromes of cosets of~$\CC_\C$; the syndromes can be generated by multiplying $P$ (in homogeneous coordinates) by elements of $\F_q^*$. Points of each orbit $\M_j$ generate $\#\M_j\cdot(q-1)$ syndromes corresponding to a set of cosets with the same weight distribution.

Every $\C$-point gives rise to $q-1$ syndromes of weight $1$ cosets. Every RC-point generates $q-1$ syndromes of weight $2$ cosets. Finally, every point on a tangent (i.e. either T-, or $(q+1)_\Gamma$-, or TO-point) as well as every IC-point gives rise to $q-1$ syndromes of weight $3$ cosets. These cases are exhaustive for our goals as every point off $\C$ lies on exactly one chord of $\C$.
\end{remark}

\begin{theorem}\label{th3_3_coset}
 Let $q\equiv\xi\pmod3$.
The distribution of weight $3$ vectors in all the cosets of the code $\CC_\C$ is given by Table \emph{\ref{tab1}}.
\begin{table*}[htbp]
\centering
   \caption{The distribution of weight $3$ words in the cosets of the $[q+1,q-3,5]_q3$  GDRS code $\CC_\C$ associated with the twisted cubic $\C$ of \eqref{eq2_cubic}, $q\equiv\xi\pmod3$. $\{\text{RC-p}\}=\{\text{RC-points}\}$, $\{\text{IC-p}\}=\{\text{IC-points}\}$}
   $
   \begin{array}{@{}c|@{}c@{}|@{}c@{}|c@{}|c@{}|c@{}|@{}c@{\,}|@{}c@{}}   \hline
&  & \text{Coset}& \text{Coset}&\text{The number }&\text{The number}&\text{Orbits}&\\
  \text{no.} &   \xi &  \text{leader}&\text{nota-}&B_3(\V)\text{ of weight}&\text{of cosets of}&\text{generating}&B_3(\V)\\
    &   &  \text{weight}&\text{tion}&\text{3 vectors}&\text{the given type}&\text{coset syndromes}&\\
    &&&&\text{in a coset}&&\\\hline
  1&   \text{any}&0&\CC_\C&0&1&\\ \hline
  2&   \text{any}&1&\V^{(1)}&0&(q+1)(q-1)&\M_1=\{\C\text{-points}\}\\ \hline
3&     1&2&\V^{(2)}_a&\frac{1}{6}(q^2-5q+4)\vphantom{H_{H_{H_H}}}&\frac{1}{3}(q^3-q)(q-1)\vphantom{H^{H^{H^H}}}&\M_5=\{\text{RC-p}\}\setminus\M_3&r_{35}^{(1)}-(q-1)\\
4&     1&2&\V^{(2)}_b&\frac{1}{6}(q^2-5q+10)\vphantom{H_{H_{H_H}}}&\frac{1}{6}(q^3-q)(q-1)\vphantom{H^{H^{H^H}}}&\M_3=\{\text{RC-p}\}\setminus\M_5&r_{33}^{(1)}-(q-1)\\
5&1&3&\V^{(3)}_a&\frac{1}{6}(q^2-3q+2)&(q^2+q)(q-1)\vphantom{H^{H^{H^H}}}&\M_2=\{\text{T-points}\}&r_{32}^{(1)}\\
6&1&3&\V^{(3)}_b&\frac{1}{6}(q^2-q)\vphantom{H_{H_{H_H}}}&\frac{1}{2}(q^3-q)(q-1)\vphantom{H^{H^{H^H}}}&\M_4=\{\text{IC-points}\}&r_{34}^{(1)}\\\hline
 7&    -1&2&\V^{(2)}&\frac{1}{6}(q^2-5q+6)\vphantom{H_{H_{H_H}}}&\frac{1}{2}(q^3-q)(q-1)\vphantom{H^{H^{H^H}}}&\M_4=\{\text{RC-points}\}&r_{34}^{(-1)}-(q-1)\\
8&-1&3&\V^{(3)}_a&\frac{1}{6}(q^2-3q+2)&(q^2+q)(q-1)\vphantom{H^{H^{H^H}}}&\M_2=\{\text{T-points}\}&r_{32}^{(-1)}\\
9&-1&3&\V^{(3)}_b&\frac{1}{6}(q^2-q-2)\vphantom{H_{H_{H_H}}}&\frac{1}{3}(q^3-q)(q-1)\vphantom{H^{H^{H^H}}}&\M_5=\{\text{IC-p}\}\setminus\M_3&r_{35}^{(-1)}\\
10&-1&3&\V^{(3)}_c&\frac{1}{6}(q^2-q+4)\vphantom{H_{H_{H_H}}}&\frac{1}{6}(q^3-q)(q-1)\vphantom{H^{H^{H^H}}}&\M_3=\{\text{IC-p}\}\setminus\M_5&r_{33}^{(-1)}\\\hline
 11&    0&2&\V^{(2)}&\frac{1}{6}(q^2-5q+6)\vphantom{H_{H_{H_H}}}&\frac{1}{2}(q^3-q)(q-1)\vphantom{H^{H^{H^H}}}&\M_4=\{\text{RC-points}\}&r_{34}^{(0)}-(q-1)\\
12&0&3&\V^{(3)}_a&\frac{1}{6}(q^2-3q)\vphantom{H_{H_{H_H}}}&(q^2-1)(q-1)\vphantom{H^{H^{H^H}}}&\M_3=\{\text{TO-points}\}&r_{33}^{(0)}\\
13&0&3&\V^{(3)}_b&\frac{1}{6}(q^2-q)\vphantom{H_{H_{H_H}}}&\left((q+1)+\frac{q^3-q}{2}\right)(q-1)=\vphantom{H^{H^{H^H}}}&\M_2\cup\M_5=&r_{32}^{(0)}=r_{35}^{(0)}\\
&&&&\vphantom{H_{H_{H_H}}}&\frac{1}{2}(q+1)(q^2-q+2)(q-1)&\{(q+1)_\Gamma\text{-points}\}\cup&\\
&&&&&&\{\text{IC-points}\}&\\\hline
   \end{array}
   $
 \label{tab1}
 \end{table*}
\end{theorem}
\begin{proof}
 The columns of the parity check matrix $H(\CC_\C)$ are points of $\C$ in homogeneous coordinates.  By Remark \ref{rem3}, the number of cosets in Table \ref{tab1} is equal to the cardinality of the corresponding orbits $\M_j$ multiplied by $\#\F_q^*=q-1$. The cardinalities $\#\M_j$ are noted in \eqref{eq2_point_orbits_gen} and  \eqref{eq2_point_orbits_j=0}.

\textbf{(1)} Let a coset leader weight be 0.

The coset is the code $\CC_\C$ of minimum distance 5; it does not contain weight 3 words. We obtain row 1 of Table \ref{tab1}.

\textbf{(2)} Let a coset leader weight be 1.

In this case, the syndrome of the coset is a column of $H(\CC_\C)$ multiplied by an element of $\F_q^*$; it is a $\C$-point. This type of cosets does not contain weight 3 vectors as minimum distance of $\CC_\C$ is equal to~5. We have $\#\M_1=q+1$, see \eqref{eq2_point_orbits_gen}; therefore the number of cosets is $\#\M_1\cdot(q-1)=(q+1)(q-1)$. We obtain row 2 of Table  \ref{tab1}.

\textbf{(3)} Let a coset leader weight be 2.

In this case, the syndrome of the coset can be represented as  a linear combination of two columns of $H(\CC_\C)$; it is an RC-point, see Remark \ref{rem2:r_ij}. Since every RC-point lies in a few $3_\C$-planes, the syndrome also can be represented (in a few manners) as  a linear combination of three columns of $H(\CC_\C)$; such a combination gives a weight 3 vector in the  coset. Thus, a weight 3 vector in the  coset corresponds to an \emph{intersection off $\C$} of a real chord and a $3_\C$-plane. If a real chord lies in a $3_\C$-plane this does not affect to the number of the weight 3 vectors. To calculate the number $\TT$ of the intersections, we should take the total number of $3_\C$-planes through an RC-point and subtract from it the number $q-1$ of $3_\C$-planes in which a real chord lies (see Lemma \ref{lem2_3 2}).

If $\xi=1$, the syndromes are points of $\M_3\cup\M_5=\{\text{RC-points}\}$, see \eqref{eq2_=1_orbit_point}. We take $r_{3,5}^{(1)}=\frac{1}{6}(q^2+q-2)$ and $r_{3,3}^{(1)}=\frac{1}{6}(q^2+q+4)$ from  Theorem \ref{Th2_rij}(i), subtract $q-1$, and obtain the values $\TT=\frac{1}{6}(q^2-5q+4)$ and $\TT=\frac{1}{6}(q^2-5q+10)$ for rows 3 and~4 of Table  \ref{tab1}, respectively; cf. the values of $B_3=B_3(\V)$ for the cosets $\V^{(2)}_a$ and $\V^{(2)}_b$ with $q=7$ in Example \ref{ex4} below. The number of the corresponding cosets is $\#\M_5\cdot(q-1)=\frac{1}{3}(q^3-q)(q-1)$ and $\#\M_3\cdot(q-1)=\frac{1}{6}(q^3-q)(q-1)$, respectively. The values $\#\M_5$ and $\#\M_3$ are taken from~\eqref{eq2_point_orbits_gen}.

If $\xi\in\{-1,0\}$, the syndromes are points of $\M_4=\{\text{RC-points}\}$, see \eqref{eq2_=1_orbit_point}, \eqref{eq2_point_orbits_j=0}.  Therefore, the total number of $3_\C$-planes through an RC-point is the value $r_{34}^{(-1)}=r_{34}^{(0)}=\frac{1}{6}(q^2+q)$ of Theorem~\ref{Th2_rij} whence $\TT=\frac{1}{6}(q^2+q)-(q-1)=\frac{1}{6}(q^2-5q+6)$. The number of the cosets is $\#\M_4\cdot(q-1)=\frac{1}{2}(q^3-q)(q-1)$, see \eqref{eq2_point_orbits_gen}. We obtain rows 7 and 11 of Table \ref{tab1}; cf. the values of $B_3=B_3(\V)$ for the cosets $\V^{(2)}$ with $q=5,8,9,11$ in Example \ref{ex4}.

\textbf{(4)} Let a coset leader weight be 3.

In this case, as every point off $\C$ lies on exactly one chord of $\C$ (see Theorem \ref{th2_HirsCor4,5}(iv)), the syndrome of the coset belongs to $\PG(3,q)\setminus(\{\C\text{-points}\}\cup\{\text{RC-points}\})$. In other words, the syndrome belongs to  $\{\text{T-points}\}\cup\{\text{IC-points}\}$ for $\xi\ne0$ (see Theorem~\ref{th2_HirsCor4,5}(ii)) or $\{\text{TO-points}\}\cup\{(q+1)_\Gamma\text{-points}\}\cup\{\text{IC-points}\}$ for $\xi=0$ (see Theorem~\ref{th2_HirsCor4,5}(iii)); see also Remarks \ref{rem2:afterNot1} and \ref{rem2:r_ij}.   Thus, a weight 3 vector in the  coset corresponds to an intersection of a tangent or an imaginary chord with a $3_\C$-plane. No  tangent or imaginary chord lies in a $3_\C$-plane, otherwise we have an intersection of a tangent or an imaginary chord with a real chord lying in the $3_\C$-plane, contradiction, see Theorem \ref{th2_HirsCor4,5}(iv). Therefore, to calculate the number of the intersections, we directly use the needed values of $r_{3j}^{(\xi)}$ from Theorem \ref{Th2_rij} without subtracting.

If $\xi\in\{-1,1\}$, we have $\M_2=\{\text{T-points}\}$, see \eqref{eq2_point_orbits_gen}. For  rows 5 and 8 of Table  \ref{tab1}, we take $r_{32}^{(1)}=r_{32}^{(-1)}=\frac{1}{6}(q^2-3q+2)$
from Theorem \ref{Th2_rij}; cf. the values of $B_3=B_3(\V)$ for the cosets $\V^{(3)}_a$ with $q=5,7,8,11$ in Example \ref{ex4}. The number of the cosets is $\#\M_2\cdot(q-1)=q(q+1)(q-1)$.

If $\xi=-1$, we have $\M_3\cup\M_5=\{\text{IC-points}\}$, $\M_3=\{3_\Gamma\text{-points}\}$, $\M_5=\{0_\Gamma\text{-points}\}$, see~\eqref{eq2_=1_orbit_point}, \eqref{eq2_point_orbits_gen}. For rows 9 and 10 of Table  \ref{tab1}, we take, respectively, $r_{35}^{(-1)}=\frac{1}{6}(q^2-q-2)$ and $r_{33}^{(-1)}=\frac{1}{6}(q^2-q+4)$
from Theorem \ref{Th2_rij}; cf. the values of $B_3=B_3(\V)$ for the cosets $\V^{(3)}_b$ and $\V^{(3)}_c$ with $q=5,8,11$ in Example \ref{ex4}. The number of the cosets is $\#\M_5\cdot(q-1)=\frac{1}{3}(q^3-q)(q-1)$ and $\#\M_3\cdot(q-1)=\frac{1}{6}(q^3-q)(q-1)$, respectively.

If $\xi=1$, we have $\M_4=\{\text{IC-points}\}$, see \eqref{eq2_=1_orbit_point}. For  row 6 of Table \ref{tab1}, we take $r_{34}^{(1)}=\frac{1}{6}(q^2-q)$ from Theorem \ref{Th2_rij}; cf. the values of $B_3=B_3(\V)$ for the coset $\V^{(3)}_b$   with $q=7$ in Example \ref{ex4}. The number of the cosets is $\#\M_4\cdot(q-1)=\frac{1}{2}(q^3-q)(q-1)$.

If $\xi=0$, we have $\M_3=\{\text{TO-points}\}$, $\M_2=\{(q+1)_\Gamma\text{-points}\}$, $\M_5=\{\text{IC-points}\}$,  see~\eqref{eq2_point_orbits_j=0}. For rows 12 and 13 of Table  \ref{tab1}, we take, respectively, $r_{33}^{(0)}=\frac{1}{6}(q^2-3q)$ and $r_{32}^{(0)}=r_{35}^{(0)}=\frac{1}{6}(q^2-q)$
from Theorem \ref{Th2_rij}(ii); cf. the values of $B_3=B_3(\V)$ for the cosets $\V^{(3)}_a$ and $\V^{(3)}_b$   with $q=9$ in Example \ref{ex4}. The number of the cosets is $\#\M_3\cdot(q-1)=(q^2-1)(q-1)$ and $(\#\M_2+\#\M_5)\cdot(q-1)=((q+1)+(q^3-q)/2)\cdot(q-1)$, respectively.
\end{proof}

The following lemma is obvious.
\begin{lemma}\label{lem3_w012W123}
Let $\CC$ be an $[n,n-4,5]_q3$ MDS code. Let $\V^{(W)}$ be one of its cosets of weight $W\in\{1,2,3\}$. Then
  \begin{align}\label{eq3:begin}
    &B_0(\V^{(1)})=B_2(\V^{(1)})=B_3(\V^{(1)})=0,~B_1(\V^{(1)})=1,~ \displaybreak[3] \\
    &B_0(\V^{(2)})=B_1(\V^{(2)})=0,~B_2(\V^{(2)})=1,~ B_0(\V^{(3)})=B_1(\V^{(3)})=B_2(\V^{(3)})=0.\notag
    \end{align}
\end{lemma}

\begin{theorem}\label{th3_coset_distr_isknown}
The weight distribution of any coset of  $\CC_\C$ is uniquely determined by Table~\emph{\ref{tab1}}.
\end{theorem}

\begin{proof}
  The code $\CC_\C^\bot$ dual to $\CC_\C$ is a $[q+1,4,q-2]_qR$ MDS code, $R=q-4$. By \cite[Th. 10]{EzerGrasSoleMDS2011}, $\CC_\C^\bot$ has $4$ distinct nonzero weights, i.e.  $s(\CC_\C^\bot)=4$. By Lemma \ref{lem3_w012W123} and Theorem \ref{th3_3_coset} we know the numbers of vectors of weights $1$, $2$, and $3=s(\CC_\C^\bot)-1$ for all the cosets of $\CC_\C$, see Table \ref{tab1} for weight 3 vectors.  Now the assertion follows from Theorem \ref{th2_HP_coset}(ii).
\end{proof}

\begin{theorem}[\emph{\textbf{classification of the cosets}}]\label{th3_disatinct_cosets}
Let $q\equiv\xi\pmod3$. For the code $\CC_\C$,  the following holds:
\begin{description}
  \item[(i)]  All weight $1$ cosets have the same weight distribution.
  \item[(ii)] If $\xi\neq1$, all weight $2$ cosets have the same weight distribution.

  For $\xi=1$, there are $2$ distinct weight distributions of weight $2$ cosets. The numbers of the corresponding cosets are $\frac{1}{3}(q^3-q)(q-1)$ and $\frac{1}{6}(q^3-q)(q-1)$.
   \item[(iii)] For $\xi=1$, the are $2$ distinct weight distributions of weight $3$ cosets. The numbers of the corresponding cosets are $(q^2+q)(q-1)$ and $\frac{1}{2}(q^3-q)(q-1)$.

    For  $\xi=-1$, there are $3$ distinct weight distributions of weight $3$ cosets. The numbers of the corresponding cosets are $(q^2+q)(q-1)$, $\frac{1}{3}(q^3-q)(q-1)$, and $\frac{1}{6}(q^3-q)(q-1)$.

       For $\xi=0$, there are $2$ distinct weight distributions of weight $3$ cosets. The numbers of the corresponding cosets are $(q^2-1)(q-1)$ and $\frac{1}{2}(q+1)(q^2-q+2)(q-1)$.
\end{description}
\end{theorem}

\begin{proof}
 The assertion directly follows from Table  \ref{tab1} and Theorem \ref{th3_coset_distr_isknown}.
\end{proof}

\section{On the weight distribution of the cosets of the code~$\CC_\C$}\label{sec_wd_coset}
We use the following well-known combinatorial identities, see e.g. \cite[Sec. 1, Eqs. (I), (III), (IV), Problem 9(a),(b)]{Riordan}:
\begin{align}
 &\binom{h}{\ell}=\binom{h}{h-\ell};\label{eq4_Riordan_ident0}\displaybreak[3]\\
 &\binom{h}{\ell}=\binom{h-1}{\ell}+\binom{h-1}{\ell-1};\label{eq4_Riordan_ident1}\displaybreak[3]\\
 &\binom{h}{m}\binom{m}{p}=\binom{h}{p}\binom{h-p}{m-p}=\binom{h}{m-p}\binom{h-m+p}{p};
 \label{eq4_Riordan_ident2}\displaybreak[3]\\
 &\sum_{j=0}^m(-1)^j\binom{h}{j}=(-1)^m\binom{h-1}{m}=(-1)^m\binom{h-1}{h-1-m}=\binom{m-h}{m};\label{eq4_Riordan_ident3}\displaybreak[3]\\
 &\sum_{j=0}^{h-m}(-1)^j\binom{h}{m+j}=\binom{h-1}{m-1}.\label{eq4_Riordan_ident4}
\end{align}
Also, it is well known, see e.g. \cite[Th. 11.3.6]{MWS}, \cite[Th. 7.4.1]{HufPless}, that for $w\ge d$, the weight distribution $A_{w}(\CC)$ of an $[n,k,d=n-k+1]_q$ MDS code $\CC$ has the form
\begin{align}\label{eq4:wd_MDS}
 A_{w}(\CC)=\binom{n}{w}\sum_{j=0}^{w-d}(-1)^j\binom{w}{j}(q^{w-d+1-j}-1).
\end{align}

\begin{lemma} Let $\CC$ be a $[q+1,q-3,5]_q$ MDS code. Then the following holds:
  \begin{align}\label{eq4:1stpartBon}
   \binom{q+1}{w}\sum_{j=0}^{w-4}(-1)^j\binom{w}{j}q^{w-4-j}=A_w(\CC)+(-1)^{w}\binom{q+1}{w}\binom{w-1}{3},~w\ge4.
  \end{align}
\end{lemma}

\begin{proof}
Using \eqref{eq4:wd_MDS} and \eqref{eq4_Riordan_ident0}, we write the left part of the assertion as
\begin{align*}
&\binom{q+1}{w}\sum_{j=0}^{w-5}(-1)^j\binom{w}{j}\left(q^{w-4-j}-1+1\right)+
\binom{q+1}{w}\cdot(-1)^{w-4}\binom{w}{w-4}\displaybreak[3]\\
&=A_w(\CC)+\binom{q+1}{w}\left((-1)^{w}\binom{w}{4}+\sum_{j=0}^{w-5}(-1)^j\binom{w}{j}\right)\displaybreak[3]\\
&=A_w(\CC)+\binom{q+1}{w}\left((-1)^{w}\binom{w}{4}+(-1)^{w-5}\binom{w-1}{w-5}\right)
\end{align*}
where we applied the 1-st equality of \eqref{eq4_Riordan_ident3} to $\sum_{j=0}^{w-5}(-1)^j\binom{w}{j}$.  Now, using \eqref{eq4_Riordan_ident0} and   \eqref{eq4_Riordan_ident1}, we obtain $\binom{w}{4}-\binom{w-1}{4}=\binom{w-1}{3}$.
\end{proof}

\begin{theorem}\label{th4:wdW1}
Let $\CC$ be the $[q+1,q-3,5]_q3$ MDS code $\CC_\C$.
   All  the weight $1$ cosets $\V^{(1)}$ of $\CC_\C$ have the same weight distribution such that, in addition to \eqref{eq3:begin}, we have
 \begin{align}
& B_4(\V^{(1)})=\binom{q}{4};\displaybreak[3]\notag\\
 &  B_w(\V^{(1)})=A_w(\CC)+(-1)^{w}\left(\binom{q+1}{w}\binom{w-1}{3}-\binom{q}{w-1}\binom{w-2}{2}\right),~w\ge5.\label{eq4:wdW1}
 \end{align}
\end{theorem}

\begin{proof} By  \eqref{eq3:begin}, in $\BB_{w}(\V^{(1)})$ of \eqref{eq2_Bw2Bon} with $d=5$, all non-zero terms correspond to $B_1(\V^{(1)})=1$. Therefore, we save in
\eqref{eq2_Bw2Bon} only summands with $v=1$. If $j=w$, no terms with $v=1$ exist. Putting $n=q+1$, we have
\begin{align*}
\BB_{w}(\V^{(1)})=\sum_{j=w-3}^{w-1}(-1)^j\binom{j+q+1-w}{j}\binom{q}{w-j-1}.
\end{align*}
Now, using consequently \eqref{eq4_Riordan_ident0} and the 2-nd equality of \eqref{eq4_Riordan_ident2}, we have
      \begin{align*}
     &\BB_{w}(\V^{(1)})=\sum_{j=w-3}^{w-1}(-1)^j\binom{j+q+1-w}{j}\binom{q}{j+q+1-w}\displaybreak[3]\\
     &=\binom{q}{q+1-w}\sum_{j=w-3}^{w-1}(-1)^j\binom{w-1}{j}.
     \end{align*}
     Changing the variable $i=j-w+3$, again using \eqref{eq4_Riordan_ident0}, and appling \eqref{eq4_Riordan_ident4} to the sum\\ $\sum_{i=0}^2(-1)^i\binom{w-1}{w-3+i}$, we obtain
      \begin{align*}
     \BB_{w}(\V^{(1)})=(-1)^{w-3}\binom{q}{w-1}\sum_{i=0}^2(-1)^i\binom{w-1}{w-3+i}=-(-1)^{w}\binom{q}{w-1}\binom{w-2}{w-4},
     \end{align*}
     whence, together with \eqref{eq2_wd_cosetBon}, \eqref{eq4:1stpartBon},  and \eqref{eq4_Riordan_ident0}, the assertion for $w\ge5$ follows. For $B_4(\V^{(1)})$, we take into account that $A_4(\CC)=0$ and use \eqref{eq4_Riordan_ident1}.
\end{proof}

\begin{theorem}\label{th4:wdW2}
Let $\CC$ be the $[q+1,q-3,5]_q3$ MDS code $\CC_\C$ for which the distribution of weight $3$ vectors in all the cosets  is given by Table \emph{\ref{tab1}}. Then, in addition to \eqref{eq3:begin}, a weight $2$ coset $\V^{(2)}$ of $\CC_\C$ has the following weight distribution:
  \begin{align}
  & B_3(\V^{(2)})=\frac{1}{3}\binom{q-2}{2}\text{ if }q\not\equiv1\pmod3;\displaybreak[3]\notag\\
  &B_3(\V^{(2)})\in\left\{\frac{q^2-5q+4}{6},~\frac{q^2-5q+10}{6}\right\}\text{ if }q\equiv1\pmod3;\notag\\
  &B_4(\V^{(2)})=\binom{q+1}{4}-\binom{q-1}{2}-(q-2)B_3(\V^{(2)});\displaybreak[3]\notag\\
    &B_w(\V^{(2)})=A_w(\CC)\displaybreak[3]\notag\\
    &+(-1)^{w}\left(\binom{q+1}{w}\binom{w-1}{3}-(w-3)\binom{q-1}{w-2}-\binom{q-2}{w-3}B_3(\V^{(2)})\right),~w\ge5.\label{eq4:wdW2}
  \end{align}
\end{theorem}

\begin{proof}
 By \eqref{eq3:begin}, in $\BB_{w}(\V^{(2)})$ of \eqref{eq2_Bw2Bon} with $d=5$, non-zero terms correspond to $B_2(\V^{(2)})=1$ and
 $B_3(\V^{(2)})$ given by Table \ref{tab1}, see also the 1-st and 2-nd rows of \eqref{eq4:wdW2}. Therefore, we save in~\eqref{eq2_Bw2Bon} only summands with $v=2,3$. The terms with $v=2$ (resp. $v=3$) exist if $j=w-3,w-2$ (resp. $j=w-3$). Putting $n=q+1$, we have
 \begin{align*}
     &\BB_{w}(\V^{(2)})=\sum_{j=w-3}^{w-2}(-1)^j\binom{j+q+1-w}{j}\binom{q-1}{w-j-2}+(-1)^{w-3}\binom{q-2}{w-3}\binom{q-2}{0}B_3(\V^{(2)}).
     \end{align*}
 Now, similarly to Proof of Theorem \ref{th4:wdW1}, we consequently use \eqref{eq4_Riordan_ident0} and the 2-nd equality of~\eqref{eq4_Riordan_ident2}, then change the variable $i=j-w+3$, again use \eqref{eq4_Riordan_ident0}, and calculate $\sum_{i=0}^1(-1)^i\binom{w-1}{w-3+i}=w-3$. As a result,
      \begin{align*}
     &\BB_{w}(\V^{(2)})=\sum_{j=w-3}^{w-2}(-1)^j\binom{j+q+1-w}{j}\binom{q-1}{j+q+1-w}-(-1)^{w}\binom{q-2}{w-3}B_3(\V^{(2)})\displaybreak[3]\\
     &=\binom{q-1}{q+1-w}\sum_{j=w-3}^{w-2}(-1)^j\binom{w-2}{j}-(-1)^{w}\binom{q-2}{w-3}B_3(\V^{(2)})\displaybreak[3]\\
     &=(-1)^{w-3}\binom{q-1}{w-2}\sum_{i=0}^1(-1)^i\binom{w-2}{w-3+i}-(-1)^{w}\binom{q-2}{w-3}B_3(\V^{(2)})\displaybreak[3]\\
     &=(-1)^{w-3}\binom{q-1}{w-2}(w-3)-(-1)^{w}\binom{q-2}{w-3}B_3(\V^{(2)}).
     \end{align*}
     Now, the assertion for $w\ge5$ can be obtained using \eqref{eq2_wd_cosetBon}, \eqref{eq4:1stpartBon}. For $B_4(\V^{(2)})$, we apply $A_4(\CC)=0$.
\end{proof}

\begin{theorem}\label{th4:wdW3}
Let $\CC$ be the $[q+1,q-3,5]_q3$ MDS code $\CC_\C$ for which the distribution of weight $3$ vectors in all the cosets  is given by Table \emph{\ref{tab1}}. Then, in addition to \eqref{eq3:begin}, a weight $3$ coset $\V^{(3)}$ of $\CC_\C$ has the following weight distribution:
  \begin{align}\label{eq4:wdW3}
  & B_3(\V^{(3)})\in\left\{\begin{array}{ccc}
                             \{(q^2-3q+2)/6,~(q^2-q)/6\} & \text{ if } & q\equiv1\pmod3 \\
                             \{\frac{q^2-3q+2}{6},~\frac{q^2-q-2}{6},~\frac{q^2-q+4}{6}\} & \text{ if } & q\equiv-1\pmod3 \\
                             \{(q^2-3q)/6,~(q^2-q)/6\} & \text{ if } & q\equiv0\pmod3
                           \end{array}
  \right.;\displaybreak[3]\\
  &B_4(\V^{(3)})=\binom{q+1}{4}-(q-2)B_3(\V^{(3)});\displaybreak[3]\notag\\
    &B_w(\V^{(3)})=A_w(\CC)+(-1)^{w}\left(\binom{q+1}{w}\binom{w-1}{3}-\binom{q-2}{w-3}B_3(\V^{(3)})\right),~w\ge5.\notag
  \end{align}
\end{theorem}

\begin{proof}
Similarly to Proof of Theorem \ref{th4:wdW2}, by \eqref{eq2_Bw2Bon}, \eqref{eq3:begin}, we have     \begin{align*}
     &\BB_{w}(\V^{(3)})=(-1)^{w-3}\binom{q-2}{w-3}\binom{q-2}{0}B_3(\V^{(3)})=-(-1)^{w}\binom{q-2}{w-3}B_3(\V^{(3)})
     \end{align*}
     whence, together with \eqref{eq2_wd_cosetBon} and \eqref{eq4:1stpartBon}, the assertion for $w\ge5$ follows. For $B_4(\V^{(3)})$, we take into account that $A_4(\CC)=0$.
\end{proof}

\begin{example}\label{ex4}
Examples of the weight distribution of the cosets of the code $\CC_\C$ are given in Table~\ref{tab2}. The notations of the cosets are taken from Table \ref{tab1}. The distributions are obtained by the relations \eqref{eq4:wdW1}--\eqref{eq4:wdW3} obtained in Theorems \ref{th4:wdW1}--\ref{th4:wdW3}.
\begin{table*}[htbp]
\centering
   \caption{Examples of the weight distribution of the cosets $\V^{(W)}$ of the $[q+1,q-3,5]_q3$  GDRS code $\CC_\C$.\newline
   $B_w=B_w(\V^{(W)})$, $N$ is the number of the cosets of the given type}
   $
   \begin{array}
   {@{}c@{\,}|@{}c@{}|@{\,}c@{\,}|@{\,}c@{\,}|@{\,}c@{\,}|@{\,}r@{\,}|r@{\,}|r@{\,}|r@{\,}|r@{\,}|r@{\,}|r@{\,}|r@{\,}|@{\,}r|@{\,}r@{}}
    \hline
   q&\text{coset}&B_1&B_2&B_3&B_4&B_5&B_6&B_7&B_8&B_9&B_{10}&B_{11}&B_{12}&N\\ \hline
   5&\CC_\C&0&0&0&0&24&0&&&&&&&1\\
   &\V^{(1)}&1&0&0&5&15&4&&&&&&&24\\
   &\V^{(2)}&0&1&1&6&11&6&&&&&&&240\\
   &\V^{(3)}_a&0&0&2&9&6&8&&&&&&&120\\
   &\V^{(3)}_b&0&0&3&6&9&7&&&&&&&160\vphantom{H_{H_{H_{H}}}}\\
   &\V^{(3)}_c&0&0&4&3&12&6&&&&&&&80\\\hline
   7&\CC_\C&0&0&0&0&336&336&1056&672&&&&&1\\
   &\V^{(1)}&1&0&0&35&217&490&966&692&&&&&48\\
   &\V^{(2)}_a&0&1&3&40&182&541&935&699&&&&&672\\
   &\V^{(2)}_b&0&1&4&35&192&531&940&698&&&&&336\vphantom{H_{H_{H_{H}}}}\\
   &\V^{(3)}_a&0&0&5&45&162&566&921&702&&&&&336\\
   &\V^{(3)}_b&0&0&7&35&182&546&931&700&&&&&1008\vphantom{H_{H_{H_{H_H}}}}\\\hline
   8&\CC_\C&0&0&0&0&882&1764&7812&12411&9898&&&    &1\\
   &\V^{(1)}&1&0&0&70&588&2268&7372&12606&9863&&&  &63\\
   &\V^{(2)}&0&1&5&75&523&2399&7251&12661&9853&&&  &1764 \\
   &\V^{(3)}_a&0&0&7&84&483&2464&7197&12684&9849&&&&504\\
   &\V^{(3)}_b&0&0&9&72&513&2424&7227&12672&9851&&&&1176\vphantom{H_{H_{H_{H}}}}\\
   &\V^{(3)}_c&0&0&10&66&528&2404&7242&12666&9852&&&&588\\\hline
   9&\CC_\C&0&0&0&0&2016&6720&40320&113760&205040&163584&&&1\\
   &\V^{(1)}&1&0&0&126&1386&8064&38760&114795&204669&163640&&   &80\\
   &\V^{(2)}&0&1&7&133&1267&8365&38389&115048&204577&163654&&   &2880 \\
   &\V^{(3)}_a&0&0&9&147&1197&8505&38235&115146&204543&163659&& &640\\
   &\V^{(3)}_b&0&0&12&126&1260&8400&38340&115083&204564&163656&&&2960\vphantom{H_{H_{H_{H_H}}}}\\\hline
   11&\CC_\C&0&0&0&0&7920&55440&554400&3366000&15037000&45074040&81962880&68301200&1\\
&\V^{(1)}&1&0&0&330&5742&61908&543180&3378375&15028145&45078044&81961836&68301320    & 120     \\
&\V^{(2)}&0&1&12&342&5424&63042&541080&3380763&15026408&45078837&81961628&68301344   & 6600      \\
&\V^{(3)}_a&0&0&15&360&5292&63420&540450&3381435&15025940&45079044&81961575&68301350 & 1320        \\
&\V^{(3)}_b&0&0&18&333&5400&63168&540828&3381057&15026192&45078936&81961602&68301347 & 4400         \vphantom{H_{H_{H_{H}}}}\\
&\V^{(3)}_c&0&0&19&324&5436&63084&540954&3380931&15026276&45078900&81961611&68301346 & 2200         \\\hline
      \end{array}
   $
 \label{tab2}
 \end{table*}
\end{example}

\section{The property of differences and a symmetry in the weight distributions of the cosets of the code $\CC_\C$}\label{sec_symmetry}

In the next theorem, for the code $\CC_\C$, we give the \emph{property of differences} for the cosets of equal weight having distinct weight distributions. The property asserts that, for $w\ge4$, the difference  between the $w$-th components of the distinct weight distributions is uniquely determined by the difference between  the $3$-rd components.
This property allows us to prove a symmetry of the coset weight distributions.
\begin{theorem}[\emph{\textbf{property of differences}}]\label{th5_differencesOfWeights}
 Let $W=2,3$. Let $3\le w\le q+1$. Let $\V^{(W)}_a, \V^{(W)}_b$ be two weight $W$ cosets of $\CC_\C$ with distinct weight distributions according to Table \emph{\ref{tab1}}. Then, for $4\le w\le q+1$, the difference $B_{w}(\V^{(W)}_a)-B_w(\V^{(W)}_b)$ between the $w$-th components of the distinct weight distributions is uniquely determined by the difference $B_{3}(\V^{(W)}_a)-B_{3}(\V^{(W)}_b)$ between  the $3$-rd components so that
\begin{align}\label{eq5_differencesOfWeights}
B_{w}(\V^{(W)}_a)-B_w(\V^{(W)}_b)=-(-1)^w\left[B_3(\V^{(W)}_a)-B_3(\V^{(W)}_b)\right]\binom{q-2}{w-3},~ W=2,3.
\end{align}
\end{theorem}

\begin{proof}
 In  \eqref{eq4:wdW2} and \eqref{eq4:wdW3}, the only term depending of $B_3(\V^{(W)})$ is $-(-1)^{w}\binom{q-2}{w-3}B_3(\V^{(W)})$. Therefore, the assertion follows from \eqref{eq4:wdW2} and \eqref{eq4:wdW3} for $W=2$ and $W=3$, respectively.
\end{proof}

 \begin{theorem}[\emph{\textbf{symmetry}}]\label{th5_symmetr}
Let  $w=3,\ldots,\lfloor(q+3)/2\rfloor$. Let $q\equiv\xi\pmod3$. Let the notation of the cosets be as in Table~\emph{\ref{tab1}}. There is the following symmetry in the weight distributions of the code~$\CC_\C$.
\begin{align*}
 &\textbf{\emph{(i)}}~~~W=2,\,\xi=1;~W=3,\,\xi=1;~W=3,\,\xi=0: \\
 &B_{q+4-w}(\V^{(W)}_a)-(-1)^qB_w(\V^{(W)}_a)=B_{q+4-w}(\V^{(W)}_b)-(-1)^qB_w(\V^{(W)}_b).\displaybreak[3]\\
 &\textbf{\emph{(ii)}}~~\xi=-1:\\
 &B_{q+4-w}(\V^{(3)}_a)-(-1)^qB_w(\V^{(3)}_a)=B_{q+4-w}(\V^{(3)}_b)-(-1)^qB_w(\V^{(3)}_b)\displaybreak[3]\\
 &=B_{q+4-w}(\V^{(3)}_c)-(-1)^qB_w(\V^{(3)}_c).
  \end{align*}
 \end{theorem}

 \begin{proof}
By \eqref{eq4_Riordan_ident0}, $\binom{q-2}{q+4-w-3}= \binom{q-2}{w-3}$. Also, it is obvious that
\begin{align*}
(-1)^w=\left\{\begin{array}{lcl}
                                 (-1)^{q+4-w}  & \text{if} & q \text{ is even}\\
                                -(-1)^{q+4-w} & \text{if} &  q \text{ is odd}
                               \end{array}\right..
\end{align*}
Therefore, by Theorem \ref{th5_differencesOfWeights} and \eqref{eq5_differencesOfWeights}, we have
\begin{align*}
B_{q+4-w}(\V^{(W)}_a)-B_{q+4-w}(\V^{(W)}_b)=\left\{\begin{array}{lcl}
                                 -\left[B_{w}(\V^{(W)}_a)-B_w(\V^{(W)}_b)\right]  & \text{if} & q \text{ is even}\\
                                ~~~~~B_{w}(\V^{(W)}_a)-B_w(\V^{(W)}_b) & \text{if} &  q \text{ is odd}
                               \end{array}\right..
\end{align*}
 We apply this relation in all cases of Table \ref{tab1} (see also Theorem \ref{th3_disatinct_cosets}) when $\CC_\C$ cosets of the same weight have distinct weight distributions. This implies  the assertions.
 \end{proof}

\begin{example}
 By Table \ref{tab2}, we have the following.
  \begin{align*}
&\textbf{(i)}~~ q=7,~\xi=1,~(-1)^q=-1,~q+4=11;\displaybreak[3]\\
& B_{11-4}(\V^{(2)}_a)+B_4(\V^{(2)}_a)=B_{11-4}(\V^{(2)}_b)+B_4(\V^{(2)}_b)=975;\displaybreak[3]\\
& B_{11-3}(\V^{(3)}_a)+B_3(\V^{(3)}_a)=B_{11-3}(\V^{(3)}_b)+B_3(\V^{(3)}_b)=707.\displaybreak[3]\\
&~~~~ q=9,~\xi=0,~(-1)^q=-1,~q+4=13;\displaybreak[3]\\
& B_{13-6}(\V^{(3)}_a)+B_6(\V^{(3)}_a)=B_{13-6}(\V^{(3)}_b)+B_6(\V^{(3)}_b)=46740.\displaybreak[3]\\
&\textbf{(ii)}~~ q=8,~\xi=-1,~(-1)^q=1,~q+4=12;\displaybreak[3]\\
& B_{12-5}(\V^{(3)}_a)-B_5(\V^{(5)}_a)=B_{12-5}(\V^{(3)}_b)-B_5(\V^{(3)}_b)=B_{12-5}(\V^{(3)}_c)-B_5(\V^{(3)}_c)=6714.
  \end{align*}
\end{example}

\end{document}